\documentclass[conference]{IEEEtran}
\usepackage{amsthm, amssymb, amsmath}
\usepackage{verbatim, float}
\usepackage{mathrsfs}
\usepackage{graphics}
\usepackage[usenames,dvipsnames]{pstricks}
\usepackage{epsfig}
\usepackage{pst-grad} 
\usepackage{pst-plot} 
\usepackage{cases}
\usepackage{algorithmic}
\usepackage{bm}
\usepackage[noadjust]{cite}
\usepackage[top=1.2cm, bottom=0.86in, left=0.6in, right=0.6in]{geometry}
\usepackage{url}

\usepackage{hhline}

\usepackage{etoolbox}
\makeatletter
\patchcmd{\@makecaption}
  {\scshape}
  {}
  {}
  {}
\makeatother

\usepackage[singlelinecheck=false]{caption}
\usepackage{diagbox}

\theoremstyle{plain}
\newtheorem{theorem}{Theorem}
\newtheorem{lemma}{Lemma}

\newtheorem{corollary}{Corollary}
\newtheorem{proposition}{Proposition}

\theoremstyle{definition}
\newtheorem{definition}{Definition}

\newtheorem{algorithm}{Algorithm}

\newcommand{\C}{{\mathcal C}}

\newcommand{\T}{{\mathcal T}}

\newcommand{\W}{{\mathcal W}}
\newcommand{\U}{{\mathcal U}}

\renewcommand{\S}{{\mathcal S}}

\newcommand{\ba}{{\boldsymbol a}}
\newcommand{\bb}{{\boldsymbol b}}
\newcommand{\bu}{{\boldsymbol u}}
\newcommand{\bv}{{\boldsymbol v}}

\newcommand{\by}{{\boldsymbol y}}

\newcommand{\bx}{{\boldsymbol{x}}}

\newcommand{\bcj}{{\boldsymbol c}_j}

\newcommand{\bcjs}{{\boldsymbol c}_{j^*}}
\newcommand{\bco}{{\boldsymbol c}_1}
\newcommand{\bcn}{{\boldsymbol c}_n}
\newcommand{\bct}{{\boldsymbol c}_2}

\newcommand{\bave}{b_{\mathsf{AVE}}}

\newcommand{\bbF}{{\mathbb F}}

\newcommand{\bbZ}{{\mathbb Z}}

\newcommand{\ff}{\mathbb{F}}

\newcommand{\fq}{\mathbb{F}_q}

\newcommand{\fql}{\mathbb{F}_{q^{\ell}}}
\newcommand{\fqa}{\mathbb{F}_{q^{a}}}
\newcommand{\fqb}{\mathbb{F}_{q^{b}}}
\newcommand{\fqas}{\mathbb{F}_{q^{a}}^*}
\newcommand{\fqbs}{\mathbb{F}_{q^{b}}^*}
\newcommand{\fqls}{\mathbb{F}_{q^{\ell}}^*}
\newcommand{\fqlms}{\mathbb{F}_{q^{\ell-s}}}

\newcommand{\fqm}{\mathbb{F}_{q^m}}

\newcommand{\LW}{L_{\mathcal{W}}}

\newcommand{\om}{\omega}

\newcommand{\gam}{\gamma}

\newcommand{\bal}{\bm{\alpha}}
\newcommand{\bals}{\bm{\alpha}^*}

\newcommand{\baljs}{\bm{\alpha}_{j^*}}

\newcommand{\bbe}{\bm{\beta}}

\newcommand{\bde}{\bm{\delta}}

\newcommand{\bom}{\bm{\omega}}

\newcommand{\bga}{\bm{\gamma}}
\newcommand{\bla}{\bm{\lambda}}

\newcommand{\bxi}{\bm{\xi}}

\newcommand{\bmu}{\bm{\mu}}
\newcommand{\bnu}{\bm{\nu}}

\newcommand{\tr}{\mathsf{Tr}}
\newcommand{\trqlq}{\mathsf{Tr}_{\bbF_{q^\ell}/\bbF_q}}

\newcommand{\spn}{\mathsf{span}_{\fq}}

\newcommand{\rank}{\mathsf{rank}_{\fq}}

\newcommand{\im}{{\sf{im}}}

\newcommand{\vbc}{\vec{\boldsymbol{c}}}
\newcommand{\vc}{\vec{\boldsymbol{c}}}

\newcommand{\vg}{\vec{\boldsymbol{g}}}

\newcommand{\vlam}{\vec{\boldsymbol{\lambda}}}

\newcommand{\define}{\stackrel{\mbox{\tiny $\triangle$}}{=}}

\newcommand{\rsk}{\text{RS}(A,k)}
\newcommand{\grskl}{\text{GRS}(A,k,\vec{\boldsymbol{\lambda}})}
\newcommand{\grsnkl}{\text{GRS}(A,n-k,\vec{\boldsymbol{\lambda}})}

\newcommand{\fa}{f(\boldsymbol{\alpha})}
\newcommand{\faj}{f(\boldsymbol{\alpha}_j)}

\newcommand{\gaj}{g(\boldsymbol{\alpha}_j)}

\newcommand{\gix}{g_i(x)}

\newcommand{\bw}{{\sf{bw}}}
\newcommand{\bwsi}{{\sf{bw}_{SI}}}

\newcommand{\w}{\bm{w}}

\begin{document}

\title{Repairing Reed-Solomon Codes\\ with Side Information}


\author{Thi Xinh Dinh$^{\dag*}$, Ba Thong Le$^*$, Son Hoang Dau$^\dag$, Serdar Boztas$^\dag$, Stanislav Kruglik$^\ddag$, Han Mao Kiah$^\ddag$,\\ Emanuele Viterbo**, Tuvi Etzion$^\#$, and Yeow Meng Chee$^\#$\\$^\dag$RMIT University, $^\ddag$Nanyang Technological University, $^*$Tay Nguyen University,\\ $^{**}$Monash University, $^\#$National University of Singapore}

\maketitle

\begin{abstract}
We generalize the problem of recovering a lost/erased symbol in a Reed-Solomon code to the scenario in which some \textit{side information} about the lost symbol is known. The side information is represented as a set $S$ of linearly independent combinations of the sub-symbols of the lost symbol. When $S = \varnothing$, this reduces to the standard problem of repairing a single codeword symbol. When $S$ is a set of sub-symbols of the erased one, this becomes the repair problem with partially lost/erased symbol.
We first establish that the minimum repair bandwidth depends on $|S|$ and not the content of $S$
and construct a lower bound on the repair bandwidth of a linear repair scheme with side information $S$.
We then consider the well-known subspace-polynomial repair schemes and show that their repair bandwidths can be optimized by choosing the right subspaces. 
Finally, we demonstrate several parameter regimes where the optimal bandwidths can be achieved 
for full-length Reed-Solomon codes.
\end{abstract}


\section{Introduction}
\label{sec:intro} 

To prevent data loss and increase data availability in distributed storage systems, a file is usually split into $k$ data chunks and transformed (encoded) into $n>k$ coded chunks using an erasure code, and then stored across $n$ different storage nodes. If the code is MDS~\cite{MW_S}, such a system can withstand any $n-k$ failures because the entire file can be recovered from any $k$ chunks. 
When only one node fails, which is usually the most typical case (see, e.g.~\cite{Rashmi2013}), a repair/replacement node must download enough data from other helper nodes to recover its lost chunk. In the \textit{repair-bandwidth} problem~\cite{Dimakis_etal2010,Dimakis_etal2010_survey}, one seeks to minimize the repair bandwidth, i.e. the amount of downloaded data required for a successful recovery of the lost chunk. 
A low-bandwidth repair scheme can also be used for \textit{degraded read}, in which requests for a chunk stored at an unavailable node can be served by other available nodes~\cite{KhanBurnsPlankPierceHuang2012}. 
This important problem has been extensively studied in the literature (see, e.g.~\cite{ 
EC_DSS_survey_2018} and the references therein). 

In this work we generalize the setting of the repair-bandwidth problem to accommodate \textit{side information} (see Fig.~\ref{fig:toy_example} for a toy example). 
In information theory, the concept of side information has been investigated in numerous contexts, including source coding~\cite{Wyner_1975}, channel coding~\cite{KeshetSteinbergMerhav_2008}, list decoding~\cite{Guruswami_2003}, index coding~\cite{BirkKol_2006, BarYossefBirkJayramKol_2011}, and private information retrieval~\cite{KadheBarciaHeidarzadehElRouayhebSprintson_2020}.
In the context of the repair problem, side information refers to the additional information that the repair node knows about the lost chunk while recovering it. The side information could arise due to a \textit{partial loss} of data, which means that part of the chunk is still accessible and serves as side information, or due to partial information gained from the previous communications or from other sources. The question of interest is that given the side information, what is the lowest repair bandwidth we can achieve. We refer to this as the repair-bandwidth \textit{with side information} problem. 

\begin{figure}[t]
    \centering \includegraphics[scale=0.9]{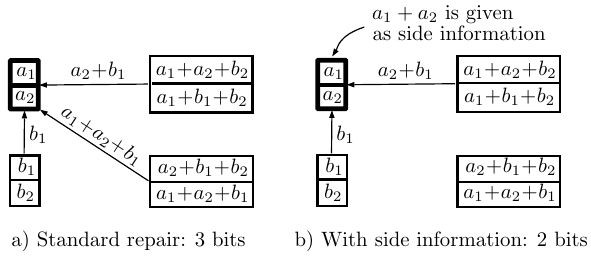}
    \caption{An illustration of repair schemes that recover $\bm{a}\hspace{-2pt} =\hspace{-2pt} (a_1, a_2)$ with and without side information. The side information $a_1\hspace{-2pt}+\hspace{-2pt}a_2$ leads to a reduction of 1 bit in the repair bandwidth. The repair node first obtains $a_2\hspace{-2pt}\leftarrow\hspace{-2pt} (a_2\hspace{-2pt}+\hspace{-2pt}b_1)\hspace{-2pt}-\hspace{-2pt}b_1$, and then $a_1 \leftarrow (a_1+a_2)-a_2$.}
    \label{fig:toy_example}
\end{figure}

In the scope of this work, we focus on Reed-Solomon codes, which is currently the most widely used families of erasure codes in distributed storage systems (see~\cite{DinhNguyenMohanSerdarLuongDau_ISIT2022}). 
The repair bandwidth as well as the closely related metric called I/O cost and sub-packetization size have been investigated in a number of recent works for different families of Reed-Solomon codes~\cite{Shanmugam2014,GuruswamiWootters2016,GuruswamiWootters2017,DauMilenkovic2017,DauDuursmaKiahMilenkovicTwoErasures2017,DauDuursmaKiahMilenkovic2018,YeBarg_TIT2017,LiWangJafarkhani-Allerton-2017,LiWangJafarkhani_TIT_2019,LiWangJafarkhani_CommLett_2021, ChowdhuryVardy2017,ChowdhuryVardy_TIT_2019,TamoYeBarg2017,TamoYeBarg2018,DauViterbo-ITW-2018,DauDuursmaChu-ISIT-2018,DuursmaDau2017,DauDinhKiahTranMilenkovic2021,LiDauWangJafarkhaniViterbo_ISIT2019,DinhNguyenMohanSerdarLuongDau_ISIT2022,DinhBoztasDauViterbo_2023,XuZhangWangZhang_2023,BermanBuzagloDorShanyTamo_ISIT_2021,ConTamo_2022,ConShuttyTamoWootters_2023}. For an $[n,k]_{\fql}$ Reed-Solomon code over the finite field $\fql$, a coded chunk, which is an element in $\fql$ (called a \textit{symbol}), is repaired from a number of $\fq$-elements (called \textit{sub-symbols}) extracted from other coded chunks. Each symbol $\bga\in \fql$ consists of $\ell$ sub-symbols in $\fq$, i.e. $\bga = (\gam_1,\ldots,\gam_\ell)\in \fq^\ell$. The repair bandwidth is measured in the number of extracted sub-symbols, while the side information of a symbol $\bga$ can be represented as a set $S$ of $\fq$-linearly independent combinations of its sub-symbols $\gam_1,\ldots,\gam_\ell$.     
In summary, our contributions are given below.
\begin{itemize}
    \item We show that the minimum repair bandwidth for a codeword symbol of a Reed-Solomon code given the side information $S$ depends on it size $|S|$, and independent of the specific choice of its elements.
    \item We obtain a lower bound on the repair bandwidth of a linear repair scheme for a failed node with side information.
    \item For subspace-polynomial repair schemes for $[n,k]_{\fql}$ Reed-Solomon codes with $n\hspace{-2pt}-\hspace{-2pt}k\hspace{-2pt}\geq\hspace{-2pt} q^m$, $m \hspace{-2pt}<\hspace{-2pt} \ell$ (\cite{GuruswamiWootters2016,GuruswamiWootters2017,DauMilenkovic2017,DauDinhKiahTranMilenkovic2021}), we prove that special subspaces can be chosen to minimize the repair bandwidth among all such schemes. 
    \item A subspace that minimizes the repair bandwidth among all subspace-polynomial repair schemes can be found by solving an optimization problem on subspace intersections, which is of its own interest. We solve the problem in a few parameter regimes, leaving others for future research.
\end{itemize}

The paper is organized as follows. Section~\ref{sec:pre} provides required notations and definitions. Section~\ref{sec:constructions} is devoted for the description and solutions of the repair with side information problem in different cases. We conclude the paper in Section~\ref{sec:conclusions}.


\section{Preliminaries}
\label{sec:pre}

\subsection{Definitions and Notations}
\label{subsec:def-notation}
\vspace{-2pt}
Let $q$ be a prime power, $\fq$ be the finite field of $q$ elements and $\fql$ be the extension field of degree $\ell$ of $\fq$.
We use $[n]$ to denote the set $\{1, 2, \ldots,n\}$, $a\mid b$ to denote that $a$ divides $b$, and $(a,b)$ for $\gcd(a,b)$, for $a,b \in \mathbb{Z}$. For a set $U$, let $U^* \define U\setminus \{0\}$, and $\bga U \define \{\bga u\colon u \in U\}$. 
\textcolor{black}{
We use $\spn(U)$ to denote the $\fq$-subspace of $\fql$ spanned by a subset $U \subseteq \fql$. We use $\dim_{\fq}(\cdot)$ and $\rank(\cdot)$ (or $\dim(\cdot)$ and $\mathsf{rank}(\cdot)$ for short) for the dimension of a subspace and the rank of a set of vectors over $\fq$. The (field) trace of an element $\bb \in \fql$ over $\fq$ is $\trqlq(\bb) = \sum_{i=0}^{\ell - 1} \bb^{q^i}$. We also write $\tr(\cdot)$ instead of $\trqlq(\cdot)$ for simplicity.}

Let $\C$ be a \emph{linear} $[n,k]$ \emph{code}  over $\fql$. Then $\C$ is an $k$-dimensional $\fql$-subspace of $\fql^n$. A \emph{codeword} of $\C$ is an element $\vbc=(\bco,\bct,\ldots,\bcn)\in \C$ and its codeword symbols are the components $\bcj$, $j \in [n]$. 
The \emph{dual} code of a code $\C$ is the orthogonal complement $\C^\perp$ of $\C$, $\C^\perp = \{\vg \in \fql^n: \langle \vc,\vg\rangle = 0, \forall \vc \in \C\}$, where $\langle \vc, \vg \rangle$ is the scalar product of $\vc$ and $\vg$. The code $\C^\perp$ is an $\fql$-subspace of $\fql^n$ with dimension $n - k$.
The elements of $\C^\perp$ are called \textit{dual codewords}. 
The number $n-k$ is called the \emph{redundancy} of the code $\C$.

\begin{definition} 
\label{def:RS}
Let $A=\{\bal_j\}_{j=1}^n$ be a subset of size $n$ in $\fql$. A \textit{Reed-Solomon code} $\rsk \subseteq \fql^n$ of dimension $k$ with \textit{evaluation points} $A$ is defined as \vspace{-5pt} 
\[
\rsk = \big\{\big(f(\bal_1),\ldots,f(\bal_n)\big) \colon f \in \fql[x],\ \deg(f) < k \big\}, \vspace{-5pt}
\]
where $\fql[x]$ is the ring of polynomials over~$\fql$.
We also use the notation RS$(n,k)$, ignoring the evaluation points. 
\end{definition}
\vspace{-5pt}
A \emph{generalized} Reed-Solomon code, $\grskl$, where $\vec{\bla} = (\bla_1,\ldots,\bla_n)\in \fql^n$, is the set of codewords $\big( \bla_1g(\bal_1),\ldots,\bla_n g(\bal_n) \big)$, where $\bla_j \neq 0$ for all $j \in [n]$, $g \in \fql[x],\ \deg(g) < n-k$.
The dual code of a Reed-Solomon code $\rsk$ is a generalized Reed-Solomon code $\grsnkl$, 
for some multiplier vector $\vec{\bla}$~(\cite[Chap.~10]{MW_S}). 
We sometimes use the notation GRS$(n,k)$, ignoring $A$ and $\vlam$.

Let $f(x)$ be a polynomial corresponding to a codeword of the Reed-Solomon code $\C=\rsk$, and $g(x)$ be a polynomial of degree at most $n-k-1$, which corresponds to a codeword of the dual code $\C^\perp$. Then
$
\sum_{j=1}^n \gaj\big(\bla_j\faj\big) = 0, 
$
and we call the polynomial $g(x)$ a \textit{check} polynomial for $\C$. 


\vspace{-3pt}
\subsection{Trace Repair Method}
\label{subsec:tracerepair}
\vspace{-3pt}

Let RS$(n,k)$ be a Reed-Solomon code over $\fql$ with evaluation points $A$,
$\vc$ ~a codeword corresponding to polynomial $f(x)$, $f \in \fql[x],\ \deg(f) < k$, and $\bcjs = f(\bals)$ is a codeword symbol/node of $\vc$, where $\bals = \baljs \in A$.
A (linear) \emph{trace repair scheme} for $f(\bals)$ corresponds to a set of $\ell$ check polynomials $\left\{g_i(x)\right\}_{i\in [\ell]}$, $g_i \in \fql[x],\ \deg(g_i) < n-k$, that satisfies the \textit{Full-Rank Condition}:
$\rank\{g_i(\bals)\}_{i\in [\ell]}=\ell$.
The repair bandwidth of such a repair scheme (in $\fq$-symbols) is ${\sf{bw}} = \sum_{\bal\in A \setminus \{\bals\}} \rank(\{g_i(\bal)\}_{i \in [\ell]})$. 
To repair all $n$ components of~$\vc$, we need $n$ such repair schemes (possibly with repetition). 
See, e.g. \cite{DauDinhKiahTranMilenkovic2021}, for a detailed explanation of why the above scheme works with an example. 



\section{Recovering an Erased Symbol with Side Information}
\label{sec:constructions}

\subsection{The Problem Description}
\label{sub:problem}

Let $\C$ be an RS$(n,k)$ code over $\fql$ with evaluation points $A\subseteq \fql$. Suppose that the codeword symbol $f(\bals)$ is erased and needs to be recovered, given a set of $\fq$-linearly independent combinations of its sub-symbols (elements of $\fq$) as \textit{side information}.
Note that for each vector of coefficients $\vec{\ba}=(a_1,\ldots,a_\ell)\in \fq^\ell$, there exists a $\bbe\in \fql$ such that the equality $\sum_{i=1}^\ell a_i\xi_i = \tr(\bbe \bxi)$ holds for all $\bxi \in \fql$. 
Therefore, we can represent the side information as a set $S = \{\bbe_i\}_{i=1}^s$, where $s = |S|$, and assume that $S$ is $\fq$-linearly independent. 
We assume that the replacement/recovery node already knows $s$ traces $\{\tr(\bbe_i f(\bals))\}_{i\in [s]}$, where $\{\bbe_i\}_{i \in [s]}\subseteq \fql$ is $\fq$-linearly independent.
Equivalently, we can represent the side information as a subspace $\S\define \spn(S)$. We call $S$ the \textit{side information set} and $\S$ the \textit{side information subspace}. Note that $S$ is a basis of $\S$ and $s=\dim(\S)$.

According to the trace-repair method, it needs to reconstruct some $\ell-s$ traces of $f(\bals)$, namely, $\{\tr(\bbe_i f(\bals))\}_{i\in [s+1,\ell]}$, referred to as \textit{target} traces, where $\{\bbe_i\}_{i \in [\ell]}$ is an $\fq$-basis of $\fql$.
We refer to $T\define \{\bbe_i\}_{i\in [s+1,\ell]}$ as the \textit{target set} and $\T\define \spn(T)$ the \textit{target subspace} with respect to the side information set $S$ (or the side information subspace $\S$). We capture this discussion in Proposition~\ref{pro:scheme_corresponds_polys}. Its proof is similar to~\cite[Thm.~4]{GuruswamiWootters2017}
and can be found in Appendix~\ref{app:linear_scheme}.

\begin{proposition}
\label{pro:scheme_corresponds_polys}
Let $S \define \{\bbe_i\}_{i\in [s]}$ be a linearly independent set and $f(\bals)$ be a symbol of Reed-Solomon code RS$(n,k)$ over $\fql$, $n \leq q^\ell$. A linear repair scheme for $f(\bals)$ with side information $S$ corresponds to a collection of $\ell-s$ polynomials $\{g_i(x)\}_{i \in [s+1,\ell]}\subset \fql[x]$, where $\deg(g_i) < n-k$, $T \define \{g_i(\bals)\}_{i\in [s+1,\ell]}$ and $S \cup T$ are linearly independent.
\end{proposition}

\subsection{Optimal Repair Bandwidths Only Depend on the Side Information Set Size}

We demonstrate below that the optimal repair bandwidth for recovering an erasure depends on $s=|S|$ but not on the specific choice of $S$. We first need an auxiliary lemma.

\begin{lemma} \label{lem:basis_completion}
Given two $\fq$-subspaces $\S$ and $\T'$ of $\fql$ of dimensions $s$ and $t=\ell-s$, respectively. Then there exists an element $\bde \in \fqls$ such that $\S\oplus \T=\fql$, where $\T = \bde \T'$.
Equivalently, given two $\fq$-linearly independent subsets $S=\{\bbe_i\}_{i\in [s]}$ and $T' = \{\bbe'_i\}_{i\in [s+1,\ell]}$ of $\fql$, where $s \in [\ell]$, there exists $\bde\in \fqls$ such that $S\cup T$, where $T = \bde T'$, forms an $\fq$-basis of $\fql$.
\end{lemma}
\begin{proof} 
For each $\bga \in \T'\setminus \{0\}$, as $\{\bde \bga \colon \bde \in \fqls\}=\fqls$, there are exactly $q^s-1$ such $\bde$ so that $\bde \bga \in \S\setminus \{0\}$ (as $|\S|=q^s$). 
Let $U \triangleq \{\bde\in \fqls\colon \exists \bga \in \T'\setminus \{0\} \text{ such that } \bde \bga \in S\}$, 
then 
$U = \bigcup_{\bga \in \T'^*} \{\bde \in \fqls\colon \bde \bga \in \S\}.$ 
We have

\[
\begin{split}
|U| &\leq \sum_{\bga \in \T'^*}|\{\bde \in \fqls \colon \bde \bga \in \S\}|\\
&= \sum_{\bga \in \T'^*} (q^s-1) = (q^{\ell-s}-1)(q^s-1)\\
&=(q^\ell-1)-(q^{\ell-s}+q^s-2) < q^\ell-1,
\end{split}
\]
for $s \in [\ell]$. Hence, there exists a $\bde \notin U$, $\bde \neq 0$, satisfying that for every $\bga \in \T'^*, \bde \bga \notin \S$. Thus, $\S \cap \bde \T' = \{0\}$ or $S \cup \bde T'$ forms a basis of $\fql$ as desired.
\end{proof} 

\begin{proposition}\label{pro:opt_bw_depends_on_size_only}
The minimum repair bandwidth for a codeword symbol 
of an RS$(n,k)$ over $\fql$ given the side information $S$ depends on $|S|$ but not on the specific choice of its elements.  
\end{proposition}
\begin{proof} 
Let $S = \{\bbe_i\}_{i\in [s]}$ and $S' = \{\bbe'_i\}_{i\in [s]}$ be two different sets of side information for repairing the same codeword symbol $f(\bals)$. It suffices to show that for every repair scheme for $f(\bals)$ with side information $S'$, there exists a repair scheme for $f(\bals)$ with side information $S$ achieving the same repair bandwidth. 
To this end, let $\{g_i(x)\}_{i \in [s+1,\ell]}\subset \fql[x]$ corresponds to the repair scheme with side information $S'$, i.e. $g_i(\bals)=\bbe'_i$, $i\in [s+1,\ell]$ and $\{\bbe'_i\}_{i\in [\ell]}$ form an $\fq$-basis of $\fql$. 
According to Lemma~\ref{lem:basis_completion}, there exists $\bde \in \fql$ such that $\{\bbe_1,\ldots,\bbe_s,\bde\bbe'_{s+1},\ldots,\bde\bbe'_\ell\}$ is linearly independent. Therefore, the polynomials $h_i(x)\triangleq \bde g_i(x)$ for all $i\in [s+1,\ell]$ form a repair scheme for $f(\bals)$ with side information $S$ and with the same repair bandwidth as $g_i(x)$'s.
\end{proof}

\subsection{A Lower Bound on the Bandwidth with Side Information}

We provide a lower bound for the repair bandwidth with side information for one erasure in a Reed-Solomon code. The lower bound is similar to those in \cite{GuruswamiWootters2016,GuruswamiWootters2017,DauMilenkovic2017,DauDinhKiahTranMilenkovic2021}, replacing $q^\ell$ by $q^{\ell-s}$ at some places. 
When $s=0$, it reduces to the existing bound (no side information).
Its proof can be found in Appendix~\ref{app:lower_bound}.

\begin{proposition}
    \label{prop:lower_bound}
    Any linear repair scheme with side information size $s$ for Reed-Solomon code RS$(A,k)$ over $\fql$ requires a repair bandwidth of at least
    \[
    t\lfloor \bave\rfloor + (n-1-t)\lceil \bave\rceil
    \]
    sub-symbols over $\fq$, where $n = |A| \leq q^\ell$, $r = n-k$, $\bave$ and $t$ are defined as 
    $\bave \define \log_q(\frac{n-1}{T})$,
    $T \define \frac{(r - 1)(q^{\ell - s} - 1) + n - 1}{q^{\ell - s}}$,
    and
    $t \define n - 1$ if $\bave \in \bbZ$, 
    $t \define \big\lfloor\frac{T - (n - 1)q^{-\lceil \bave\rceil}}{q^{-\lfloor \bave\rfloor} - q^{-\lceil \bave \rceil}}\big\rfloor$ otherwise.
\end{proposition}

In some special cases, the lower bound in Proposition~\ref{prop:lower_bound} can be explicitly computed.

\begin{corollary}
    \label{cr:lower_bound_special_case}
    Consider a full-length RS$(n=q^\ell,k)$ code over $\fql$ with $n-k = q^m$ for some $1\leq m<\ell$. Assume that $(\ell-s) \mid \ell$ and $\ell\geq m(\ell-s)$. Then every linear repair scheme with side information set size $s$ 
    requires a repair bandwidth of at least $(q^\ell-1)(\ell-s)-\frac{(q^{\ell-s}-1)(q^m-1)}{q-1}$ sub-symbols in $\fq$.
\end{corollary}
\begin{proof}
    With $n=q^\ell$ and $n-k = q^m$, we have
    \[
    T = \dfrac{(q^m-1)(q^{\ell-s}-1)+q^\ell-1}{q^{\ell-s}} \vspace{-5pt}
    \]
    and \vspace{-5pt}
    \[
    \bave\hspace{-2pt} = \hspace{-2pt}\log_q\dfrac{q^\ell-1}{T} \hspace{-2pt}=\hspace{-2pt} \log_q\hspace{-2pt}\bigg(\hspace{-2pt}\dfrac{q^\ell-1}{(q^m-1)(q^{\ell-s}-1)+q^\ell-1}q^{\ell-s}\hspace{-2pt}\bigg).
    \]
    
    Next, we show that $\ell\hspace{-2pt}-\hspace{-2pt}s\hspace{-2pt}-\hspace{-2pt}1 \hspace{-2pt}<\hspace{-2pt} \bave \hspace{-2pt}<\hspace{-2pt} \ell\hspace{-2pt}-\hspace{-2pt}s$. Indeed, the second inequality is obvious because $q^\ell\hspace{-2pt}-\hspace{-2pt}1 \hspace{-2pt}<\hspace{-2pt} (q^m\hspace{-2pt}-\hspace{-2pt}1)(q^{\ell-s}\hspace{-2pt}-\hspace{-2pt}1)\hspace{-2pt}+\hspace{-2pt}q^\ell\hspace{-2pt}-\hspace{-2pt}1$. For the first inequality, we need to show that \vspace{-5pt}
    \[
    \dfrac{q^\ell-1}{(q^m-1)(q^{\ell-s}-1)+q^\ell-1} > \dfrac{1}{q},
    \]
    which is equivalent to \vspace{-10pt}
    \[
    (q^\ell-1)(q-1) > (q^m-1)(q^{\ell-s}-1) \Longleftrightarrow \dfrac{q^\ell-1}{q^m-1}> \dfrac{q^{\ell-s}-1}{q-1},\vspace{-5pt}
    \]
    which is true because\vspace{-5pt}
    \[
    \dfrac{q^\ell-1}{q^m-1} \geq \dfrac{q^{m(\ell-s)}-1}{q^m-1} = \sum_{i=0}^{\ell-s-1} q^{mi} \geq \sum_{i=0}^{\ell-s-1} q^{i} = \dfrac{q^{\ell-s}-1}{q-1},  
    \]
    noting that either the first or the second inequality must be strict: if $m=1$ (so that the second inequality becomes equality) then the first inequality is strict since $q^{\ell} > q^{\ell-s} = q^{m(\ell-s)}$. Thus, $\bave \notin \bbZ$ and $\lfloor \bave \rfloor = \ell-s-1$ and $\lceil \bave \rceil = \ell-s$. Plugging this in the formula for $t$ in Proposition~\ref{prop:lower_bound} we obtain\vspace{-3pt}
    \[
    t = \left\lfloor\frac{T - (q^\ell - 1)q^{-\lceil \bave\rceil}}{q^{-\lfloor \bave\rfloor} - q^{-\lceil \bave \rceil}}\right\rfloor = \dfrac{(q^m-1)(q^{\ell-s}-1)}{q-1}.
    \]

    Finally, using Proposition \ref{prop:lower_bound} we obtain a lower bound of\vspace{-5pt}
    \begin{align*}
    t(\ell-s-1) + (q^\ell-1-t)(\ell-s) = (q^\ell-1)(\ell-s) - t \\
    = (q^\ell-1)(\ell-s)-(q^{\ell-s}-1)(q^m-1)/(q-1)
    \end{align*}
    sub-symbols over $\fq$ on the repair bandwidth as claimed.
\end{proof}

\subsection{Optimal Subspace-Polynomial-Based Repair Schemes}

In this section we investigate the repair bandwidth incurred by the subspace-polynomial repair scheme introduced in~\cite{DauMilenkovic2017, DauDinhKiahTranMilenkovic2021}, which generalizes the trace-polynomial-based scheme in~\cite{GuruswamiWootters2016, GuruswamiWootters2017}, under the new assumption of side information.
We show that in contrast to the standard repair problem (with no side information), the repair bandwidth of such a scheme depends on the specific choice of the subspace.
In particular, we transform the problem of finding subspace-polynomial repair schemes with minimum bandwidths possible into another one on subspace intersection, which on its own is an intriguing problem.

Before presenting Theorem~\ref{thm:subspace_intersection}, we note that given side information set $S = \{\bbe_i\}_{i\in [s]}$, to construct a subspace-polynomial repair scheme, one first needs to find a target set $T=\{\bbe_i\}_{i \in [s+1,\ell]}$ such that $S\cup T$ forms an $\fq$-basis of $\fql$ (see Proposition~\ref{pro:scheme_corresponds_polys} and the discussion preceding it). Next, given that $r = n -k \geq q^m$, for some $m< \ell$, one picks an $m$-dimensional subspace $\W$ of $\fql$, and form the $\ell-s$ check polynomials $g_i(x)\define \frac{\LW\big(\bbe_i (x-\bals)\big)}{x-\bals}$, $i\in [s+1,\ell]$. Note that $\LW(x) \define \prod_{\bom \in \W} (x-\bom)$ is the subspace polynomial which is a linearly map constructed over $\W$ and $\ker(\LW) = \W$.
The check polynomials are then used in the trace repair scheme. 

\begin{lemma}
    \label{lem:bandwidth_intersection}
    Consider an RS$(n,k)$ with evaluation points $A \subseteq \fql$ satisfying $n-k\geq q^m$, $m<\ell$. Consider also a repair scheme with side information of size $s$, that consists of $\ell-s$ polynomials $\left\{\gix\right\}_{i\in [s+1,\ell]}$, where $\gix\triangleq\LW\big(\bbe_i (x-\bals)\big)/(x-\bals)$ and $T\triangleq\{\bbe_i\}_{i\in [s+1,\ell]}$ is a target set. This scheme has bandwidth (with $|A|-1$ helper nodes) 
    \[
    (|A|-1)(\ell-s)-\textstyle\sum_{\bal \in A \setminus \{\bals\}}\dim\big((\bal-\bals)\T\cap\W\big)
    \]
    sub-symbols in $\fq$, where $\T\triangleq \spn(T)$. 
\end{lemma}
\begin{proof}
    The node storing $\fa$ computes $\ell-s$ traces $\tr\Big(\frac{\LW\big(\bbe_i (\bal-\bals)\big)}{\bal-\bals}\fa\Big)$, $i\in [s+1,\ell]$. However, due to the linearity of trace, it only needs to send $\rank\big(\left\{\LW\big(\bbe_i (\bal-\bals)\big)\right\}_{i\in [s,\ell+1]}\big)$ traces. To compute this rank, let $\U\triangleq (\bal-\bals)\T = \spn\big((\bal-\bals)T\big)$ and $\tau \colon \U \to \fql$ defined as $\tau(u)=\LW(u)$ for every $u\in \U$.
    Then $\dim(\U) = \dim(\T)=\ell-s$ and $\ker(\tau) = \U \cap \ker(\LW) = \U \cap \W$. Using the rank-nullity theorem, we obtain
    \[
    \begin{split}
    &\rank\big(\left\{\LW\big(\bbe_i (\bal-\bals)\big)\right\}_{i\in [s,\ell+1]}\big)\\ 
    &= \dim\big(\im(\tau)\big) = \dim(\U) - \dim\big(\ker(\tau)\big) \\
    &=(\ell-s) - \dim\big((\bal-\bals)\T\cap \W\big).
    \end{split}
    \]
    Summing this over all $\bal\in A\setminus \{\bals\}$ completes the lemma. 
\end{proof}

\textcolor{black}{Given the side information} subspace $\S$, when constructing a subspace-polynomial repair scheme, we have the freedom in selecting relevant target subspace $\T$ and $\W$. Hence, one can optimize the bandwidth over such $\T$ and $\W$. This is in stark contrast to the case of standard repair \textit{without} side information, in which any subspace $\W$ would lead to the same repair bandwidth~\cite{DauMilenkovic2017,DauDinhKiahTranMilenkovic2021}. Theorem~\ref{thm:subspace_intersection} formalizes this fact. 

\begin{theorem}
\label{thm:subspace_intersection} 
Consider an RS$(A,k)$, $A\subseteq \fql$, $|A| = n$, and $n-k\geq q^m$, $m<\ell$. Then the minimum bandwidth that a subspace-polynomial repair scheme for $f(\bals)$ ($\bals \in A$) can achieve, given the side information subspace $\S$, $\dim(\S)\hspace{-2pt}=\hspace{-2pt}s$, is\vspace{-5pt} 
\[
(|A|-1)(\ell-s) - \textstyle\max_{\T, \W} \sum_{\bal \in A\setminus\{\bals\}}\dim\big((\bal-\bals)\T\cap\W\big),
\]
where the $\max$ is taken over all $\fq$-subspaces $\T$ and $\W$ of $\fql$ with $\dim(\T)=\ell-s$, $\S\oplus \T=\fql$, and $\dim(\W)=m$.
\end{theorem}
\begin{proof}
Follows directly from Proposition~\ref{pro:scheme_corresponds_polys} and Lemma~\ref{lem:bandwidth_intersection}.
\end{proof} 

Note that Theorem~\ref{thm:subspace_intersection} converts the repair bandwidth with side information problem restricted to subspace-polynomial repair scheme to a pure subspace-intersection problem stated below.

\textbf{(Subspace-Intersection Problem)} Given $\bals\in A \subseteq \fql$ and an $s$-dimensional subspace $\S$ of $\fql$, find $\T$ and $\W$ that maximizes the sum $\sum_{\bal \in A\setminus\{\bals\}}\dim\big((\bal-\bals)\T\cap\W\big)$ among all $(\ell-s)$-dimensional subspaces $\T$ and $m$-dimensional subspaces $\W$ satisfying $\S\oplus\T=\fql$.

The subspace-intersection problem can be tricky to solve, especially for general $A$. Therefore, we limit ourselves to the more tractable case when $A\equiv \fql$, for which optimal repair bandwidths were known for the standard repair setting (without side information) when $q^\ell-k\geq q^m$~\cite{DauMilenkovic2017,DauDinhKiahTranMilenkovic2021}. 
We also assume that $(\ell-s)\mid \ell$. 
With these assumptions, in Corollary~\ref{cr:full_length}, we can replace $\T$ by $\fqlms$ in the optimization problem. 
Note that while $\fqlms$ may not be a valid $\T$ (i.e $\fqlms \oplus \S \neq \fql$), by Lemma~\ref{lem:basis_completion}, at least one of its cosets is.   
Finally, although this corollary provides an \textit{upper bound} instead of an exact formulation for the bandwidth, later on, using the lower bound in Proposition~\ref{prop:lower_bound}, we can establish optimal repair bandwidths for subspace polynomial schemes in some parameter regimes. 

\begin{corollary}
    \label{cr:full_length}
    Consider a full-length RS$(n=q^\ell,k)$ with evaluation points $A = \fql$, where $n-k\geq q^m$, $m<\ell$. Let $S$ be a side information set with $s=|S|$ and $(\ell-s)\mid \ell$. Then there exists a subspace-polynomial repair scheme for $f(\bals)$ ($\bals \in A$), given the side information set $S$, with repair bandwidth
    \begin{equation}
    \label{eq:upper_bound}
    (q^\ell-1)(\ell-s) - \max_{\W} \textstyle\sum_{\bga \in \fqls}\dim\big(\bga\ff_{q^{\ell-s}}\cap\W\big),
    \end{equation}
where the $\max$ is taken over all $m$-dimensional $\fq$-subspaces $\W$ of $\fql$.
\end{corollary}
\begin{proof}
    Let $\S\define \spn(S)$ be the side information subspace. By Theorem~\ref{thm:subspace_intersection}, the minimum bandwidth achieved by a subspace-polynomial repair scheme with side information subspace $\S$ is
    \begin{equation}
    \label{eq:subfield_as_target}        
    (q^\ell-1)(\ell-s) - \max_{\T, \W} \sum_{\bal \in A\setminus\{\bals\}}\dim\big((\bal-\bals)\T\cap\W\big),
    \end{equation}
    where the $\max$ is taken over all $\fq$-subspaces $\T$ and $\W$ of $\fql$ with $\dim(\T)=\ell-s$, $\S\oplus \T=\fql$, and $\dim(\W)=m$. To prove the existence of a repair scheme with bandwidth given by \eqref{eq:upper_bound}, we show that a (multiplicative) coset of $\fqlms$ can be a (valid) target subspace. Indeed, according to Lemma~\ref{lem:basis_completion}, there exists $\bde \in \fql^*$ such that $\S \oplus \bde \fqlms = \fql$, i.e. $\bde \fqlms$ is a valid target subspace w.r.t. the side information subspace $\S$. Hence, setting $\T=\bde \fqlms$ in \eqref{eq:subfield_as_target} and using the assumption that $A = \fql$, there exists a subspace-polynomial repair scheme given the side information subspace $\S$ that achieves the bandwidth
    \[
    (q^\ell-1)(\ell-s) - \max_{\W} \sum_{\bal \in \fql\setminus\{\bals\}}\dim\big((\bal-\bals)\bde \fqlms \cap\W\big),    
    \]
    which is the same as \eqref{eq:upper_bound} when replacing $(\bal-\bals)\bde$ by $\bga\in \fqls$, noting that $\big\{(\bal-\bals)\bde \colon \bal \in \fql \setminus \{\bals\}\big\} \equiv \fqls$.
\end{proof}

Using Corollary~\ref{cr:full_length}, assuming a full-length code with $(\ell-s) \mid  \ell$, we can now construct a few concrete subspace-polynomial repair schemes that achieve optimal repair bandwidths among \textit{all} linear schemes.
To construct the first repair scheme achieving optimal repair bandwidth, we first prove an auxiliary lemma.

\begin{lemma}
    \label{lem:coprime}
    For every $a\mid \ell$, $b\mid \ell$, and $(a,b)=1$, and for every $\bga, \bde \in \fqls$, it holds that $\dim(\bga\fqa \cap \bde\fqb) \in \{0,1\}$.
\end{lemma}
\begin{proof}
    Note that $\fqas = \big\{\bxi^{\frac{q^\ell-1}{q^a-1}}\big\}_{i=0}^{q^a-2}$ and $\fqbs = \big\{\bxi^{\frac{q^\ell-1}{q^b-1}}\big\}_{i=0}^{q^b-2}$, where $\bxi$ is a primitive element of $\fql$. To show that $\dim(\bga\fqa \cap \bde\fqb)\in \{0,1\}$, it suffices to show that for any $\bu, \bv \in \bga\fqas \cap \bde\fqbs$, it holds that $\bu/\bv \in \fq$. Indeed, for such $\bu$ and $\bv$, there exist $x$, $y$, $z$, and $w$ such that \vspace{-5pt}
    \[
    \begin{split}
    \bu &= \bga \big(\bxi^{\frac{q^\ell-1}{q^a-1}}\big)^x = \bde\big(\bxi^{\frac{q^\ell-1}{q^b-1}}\big)^z,\
    \bv = \bga \big(\bxi^{\frac{q^\ell-1}{q^a-1}}\big)^y = \bde\big(\bxi^{\frac{q^\ell-1}{q^b-1}}\big)^w,
    \end{split}
    \]
    which implies that\vspace{-5pt}
    \[
    \frac{\bu}{\bv} = \big(\bxi^{\frac{q^\ell-1}{q^a-1}}\big)^{x-y} = \big(\bxi^{\frac{q^\ell-1}{q^b-1}}\big)^{z-w}\in \fqa \cap \fqb = \bbF_{q^{(a,b)}} = \fq.
    \]
    The proof follows.
\end{proof}

The following theorem indicates the existence of optimal repair schemes for a full-length Reed-Solomon codes with side information size $s$, where $(\ell - s)|\ell$. We prove that the existing subspace repair scheme can be constructed from a subfield $\fqm$ of $\fql$, where $n - k \geq q^m$, $\ell > m \geq 1$, $m|\ell$, and $(\ell - s,m) = 1$. However, for any coset of $\fqm$, the proof is still right.

\begin{theorem}
    \label{thm:bandwidth_coprime}
   Consider a full-length Reed-Solomon codes RS$(n = q^\ell, k)$ over $\fql$ with $n - k \geq q^m$ for some $\ell > m \geq 1$. If $(\ell-s)\mid \ell$, $m\mid \ell$, and $(\ell - s,m) = 1$, then there exists a linear repair scheme with side information of size $s$ that uses the repair bandwidth of $(q^\ell - 1)(\ell - s) - (q^{\ell - s} -1)(q^m -1)/(q-1)$ sub-symbols in $\fq$. The scheme is optimal when $n-k=q^m$.
\end{theorem}

\begin{proof}
Set $\W \define \fqm$, and let $T \define \{\bbe_i\}_{i=s+1}^\ell$ be a basis of ~$\T \define \ff_{q^{\ell -s}}$. We now consider the subspace repair scheme constructed from $\W$.
The statement that the achieved bandwidth is optimal among all linear repair schemes is obvious due to Corollary~\ref{cr:lower_bound_special_case}, noting that the assumptions $(\ell-s) \mid \ell$, $m \mid \ell$, and $(\ell-s,m)=1$ imply $\ell \geq (\ell-s)m$. It remains to show that the stated scheme has the stated bandwidth. 
Indeed, from Lemma \ref{lem:bandwidth_intersection}, it is sufficient to prove that 
$\sum_{\bga \in \fqls}\dim\big(\bga\T \cap \W\big) = \frac{(q^{\ell - s} -1)(q^m -1)}{q-1}$.
By Lemma \ref{lem:coprime}, we note that to get the sum $\sum_{\bga \in \fqls}\dim(\bga\T \cap \fqm)$, we compute the number of elements $\bga$ so that $\dim(\bga\T \cap \fqm) = 1$.
As the set of $1$-dimensional intersections $\bga\T \cap \W$  is a partition of $\W$ into $1$-dimensional subspaces, there are $\frac{|\W^*|}{q-1}$ such subspaces. Moreover, each of them is repeated $q^{\ell-s} -1$ times, since $\bga'\T = \bga\T$ for all $\bga' \in \bga\T$. Thus, the number of elements $\bga$ with $\dim(\bga\T \cap \fqm) = 1$ is $(q^{\ell - s} -1)(q^m -1)/(q-1)$, which completes the proof. 
\end{proof}
\vspace{-5pt}

Now we consider a greedy
construction that generates $m$-dimensional subspaces $\W$, $m < \ell$, that generate subspace-polynomial repair schemes with minimal repair bandwidths. Assume that $a \mid \ell$. The aim is to construct a subspace $\W$ of $\fql$ satisfying $\dim(\bga\fqa \cap \W)\in \{0,1\}$ for all $\bga \in \fqls$. 

\vspace{-5pt}
\begin{lemma}
\label{lem:greedy_construction}
Assume that $a\mid \ell$ and that $q^\ell \hspace{-2pt} > \hspace{-2pt} \binom{q^{m-1}}{2}\Big(\frac{q^a-1}{q-1}\Big)^2+1$. Then there exists an $m$-dimensional subspace $\W$
 satisfying that $\dim(\bga\fqa \cap \W)\in \{0,1\}$ for all $\bga \in \fqls$.
 \end{lemma}
\begin{proof}
We will construct in a greedy manner a set $\{\w_1, \dots, \w_m\}\subset \fqls$ that satisfies two properties given below. 
\begin{itemize}
    \item (P1) $\{\w_1, \dots, \w_m\}$ is $\fq$-linearly independent, and 
    \item (P2) the subspace $\W_m \define \spn\big(\{\w_1, \dots, \w_m\}\big)$ satisfies that $\dim(\bga\fqa \cap \W_m)\in \{0,1\}$ for all $\bga \in \fqls$ 
\end{itemize}
The first element $\w_1$ can be picked arbitrarily in $\fqls$ because $\W_1 \define \spn(\{w_1\})$ satisfies (P1) and (P2) obviously. Assume that we have already had a set $\{\w_1, \dots, \w_{m-1}\}$ that satisfies (P1) and (P2). We now show that we can find $\w_m$ so that $\{\w_1, \dots, \w_m\}$ satisfies (P1) and (P2) given that $a \mid \ell$ and $q^\ell \hspace{-2pt} > \hspace{-2pt} \binom{q^{m-1}}{2}\Big(\frac{q^a-1}{q-1}\Big)^2+1$.
Consider the set \vspace{-5pt}
\[
B_{m-1} \define \{ \bal_1\bu + \bal_2\bv \colon \bu, \bv \in \W_{m-1}, \bu\neq \bv, \bal_1, \bal_2 \in \fqa\}.\vspace{-5pt}
\]

\noindent\textbf{Claim 1}: $|B_{m-1}| \leq \binom{q^{m-1}}{2}\Big(\frac{q^a-1}{q-1}\Big)^2+1 < |\fql|$.

\noindent\textbf{Claim 2}: Any $\w_m\in \fql\setminus B_{m-1}$ satisfies (P1)-(P2).\vspace{-5pt}

\begin{proof}[Proof of Claim 1]
Note that $\bal_1=\bal_2=0$ gives $0\in B_{m-1}$. Since $(\tau\bal) \bu = \bal (\tau\bu)$ and $\tau\bu\in \W_{m-1}$ for $\tau\in \fq$ and $\bu\in \W_{m-1}$, to count the elements of $B_{m-1}$ corresponding to $\bal_1\neq 0$ and $\bal_2\neq 0$, we only need to consider $(q^a-1)/(q-1)$ values for each $\bal_1$ and $\bal_2$. Moreover, we can ignore the case $\bal_1\neq 0$ and $\bal_2 =0$ or vice versa as the resulting elements are already counted for $\bal_1 \neq 0 $ and $\bal_2 \neq 0$ when setting either $\bv = 0$ or $\bu = 0$, respectively. Thus, other than $0$, $B_{m-1}$ has at most $\big(\frac{q^a-1}{q-1}\big)^2 \binom{q^{m-1}}{2}$ other elements, where the binomial factor counts the number of distinct pairs $\bu, \bv\in \W_{m-1}$. Thus, $|B_{m-1}|\leq \big(\frac{q^a-1}{q-1}\big)^2 \binom{q^{m-1}}{2} + 1$ elements.
\end{proof}\vspace{-5pt}
\begin{proof}[Proof of Claim 2]
Since $\W_{m-1}\subseteq B_{m-1}$, $\w_m\notin \W_{m-1}$, which implies (P1).
Assume, for the sake of contradiction, that $\dim(\W_m\cap \bga \fqa)\geq 2$ for some $\bga \in \fqls$. Then there exist $\bu, \bv\in \W_{m-1}$, $\bu\neq \bv$ so that either a) $\{\w_m+\bu,\w_m+\bv\}\subset \bga\fqa$ and $\rank(\{\w_m+\bu,\w_m+\bv\}) = 2$, or b) $\{\w_m+\bu,\bv\}\subset \bga\fqa$ and $\rank(\{\w_m+\bu,\bv\}) = 2$. If a) occurs, then there exist $\bx,\by\in \fqa$, $\bx\neq 0$, $\by\neq 0$, $\bx\neq \by$, so that $\w_m+\bu=\bga\bx$ and $\w_m+\bv=\bga\by$, which implies that $\w_m = \frac{\by}{\bx-\by}\bu + \frac{\bx}{\by-\bx}\bv \in B_{m-1}$, which contradicts our assumption. The case b) can be treated similarly. \vspace{-5pt} 
\end{proof}

The proof of Lemma~\ref{lem:greedy_construction} follows from these two claims. Indeed, by Claim~1, there exists at least one element in $\fql\setminus B_{m-1}$, which is the desired $\w_m$ according to Claim~2. 
\end{proof}


\begin{theorem}
    \label{thm:greedy_scheme}
     Consider a full-length Reed-Solomon codes RS$(n = q^\ell, k)$ over $\fql$ with $n - k \geq q^m$ for some $\ell > m \geq 1$. If $(\ell-s)\mid \ell$ and $q^\ell \hspace{-2pt} > \hspace{-2pt} \binom{q^{m-1}}{2}\Big(\frac{q^{\ell-s}-1}{q-1}\Big)^2+1$, then there exists a linear repair scheme with side information of size $s$ that uses the repair bandwidth of $(q^\ell - 1)(\ell - s) - (q^{\ell - s} -1)(q^m -1)/(q-1)$ sub-symbols in $\fq$. The scheme is optimal when $n-k=q^m$.
\end{theorem}
\begin{proof}
By Lemma \ref{lem:greedy_construction}, there exists an $m$-dimensional subspace $\W$ satisfying $\dim(\bga\fqlms \cap \W)\in \{0,1\}$, for all $\bga \in \fqls$. The rest of the proof proceeds similarly to that of Theorem~\ref{thm:bandwidth_coprime}.
\end{proof}

\subsection{Bandwidth Reductions Given Side Information}

To illustrate the repair bandwidth reduction in the presence of side information, we consider as an example the parameter regime assumed in Theorem~\ref{thm:bandwidth_coprime}. 
Case 1: $s=\ell-1$ and $m=\ell/d$ for some constant $d\geq 2$. Theorem~\ref{thm:bandwidth_coprime} gives a repair bandwidth with side information $\bwsi = (q^\ell-1)-(q^{\ell/d}-1)$. The optimal repair bandwidth with no side information is $\bw = (q^{\ell-1}-1)(1-1/d)\ell$ (see~\cite{DauMilenkovic2017,DauDinhKiahTranMilenkovic2021}). Clearly, $\lim_{\ell\to\infty}\bwsi/\bw=0$.
Case 2: $s=c\ell/(c-1)$, i.e. $\ell-s=\ell/c$, and $m=\ell/d$, for some constants $c,d\geq 2$. Then $\bwsi = (q^\ell-1)\ell/c-(q^{\ell/c}-1)(q^{\ell/d}-1)/(q-1)$, whereas $\bw = (q^\ell-1)(d-1)\ell/d$. Clearly, $\bw-\bwsi \geq (q^{\ell/c}-1)(q^{\ell/d}-1)/(q-1)\to \infty$ as $\ell\to \infty$. 

\section{Conclusions}
\label{sec:conclusions}
We proposed the problem of repairing a single erasure of Reed-Solomon codes with side information, which generalizes the standard repair problem, 
and established a lower bound on the repair bandwidth of a linear repair scheme. The problem of constructing optimal subspace-polynomial repair schemes can be reduced to a subspace intersection problem, which is interesting in its own right. We settled this problem for a few parameter regimes, leaving the general case open for future research.


\section*{Acknowledgement}
The work of Son Hoang Dau was supported by the Australia Research Council DECRA Grant DE180100768 and DP Grant DP200100731. 
The work of Han Mao Kiah was supported by the Ministry of Education, Singapore, under its MOE AcRF Tier~2 Award under Grant MOE-T2EP20121-0007 and MOE AcRF Tier~1 Award under Grant RG19/23. 
The work of Stanislav Kruglik was supported by the Ministry of Education, Singapore, under its MOE AcRF Tier~2 Award under Grant MOE-T2EP20121-0007. 
\newpage
\bibliographystyle{IEEEtran}
\bibliography{ISIT-RepairingRSCodesWithSideInfo}

\newpage
\section{Appendix}
\label{sec:appendix}

\subsection{Proof of Proposition~\ref{pro:scheme_corresponds_polys}}
\label{app:linear_scheme}

The first part of this appendix is devoted for the discussion on the definition and the existence of a linear repair scheme for a failed node with side information of size $s$ of Reed-Solomon code RS$(n,k)$. Similar to an (exact) linear repair scheme for a failed node
with standard repair, a linear repair scheme for a node with side information of size $s$ is described by elements $\bga$'s used in each trace along with a linear algorithm. 

We first propose the definition of a linear repair scheme with side information of size $s$ which is modeled after the definition of a linear repair scheme with the standard repair in \cite{GuruswamiWootters2016}.

\begin{definition}
    \label{def:linear_repair_scheme}
    A linear repair scheme with side information $S = \{\bbe_i\}_{i \in [s]}$ for a symbol $f(\bals)$ of Reed-Solomon code RS$(n,k)$ with evaluation point set $A$, $|A| = n$, over the coding field $\fql$ and the base field $\fq$ consists of\\
    $\bullet$ a set $Q_{\bal} \subset \fql$, for each $\bal\neq \bals$, and\\
    $\bullet$ $\ell-s$ coefficients $\eta_i \in \fq, i \in [s+1, \ell]$, where $\eta_i$'s are $\fq$-linear coefficients of the queries $\cup_{\bal\in A\setminus\{0\}}\{\tr(\bga f(\bal)): \bga \in Q_{\bal}\}$
    so that there is a linear reconstruction algorithm that computes 
    $f(\bals) = \sum_{i\in [\ell]} \eta_i\bnu_i$,
    where $\eta_i = \tr\big(\bbe_i f(\bals)\big)$, $i\in [s]$, which are already known from the side information $S$, and $\{\bnu_1, \dots, \bnu_{\ell}\}$ is an $\fq$-basis of $\fql$.
    The repair bandwidth $b$ is the total number of sub-symbols in $\fq$ returned by each node $\bal$, i.e.,
    $b = \sum_{\bal \in A \setminus\{\bals\}}\mid Q_{\bal} \mid$.
\end{definition}

\begin{lemma}
    \label{lem:scheme_queries_coefficients}
    Suppose there is a linear repair scheme for repairing $f(\bals)$ of RS$(A,k)$ with side information $S = \{\bbe_i\}_{i \in [s]}$ given by a set $\{Q_{\bal}\}_{\bal\in A\setminus\{\bals\}}$ and a linear algorithm as in Definition \ref{def:linear_repair_scheme}.
    Then, there is an $\fq$-linearly independent set $B = \{\bbe_i\}_{i \in [s+1,\ell]}$ so that there are elements $\bmu_{\bbe_i, \bal}$ satisfying $\bbe_i f(\bals) = \sum_{\bal \in A \setminus \{\bals\}}\bmu_{\bbe_i, \bal}f(\bal)$, for all $f \in \fql[x]$ of degree less than $k$, where $\{\bmu_{\bbe_i, \bal}\}_{\bbe_i \in B} \subseteq \spn(Q_{\bal})$.
\end{lemma}

\begin{proof}
    Suppose $B = \{\bbe_{s+1}, \dots, \bbe_{\ell}\}$ is an $\fq$-linearly independent set of $\fql$ so that $\{\bbe_1, \dots,\bbe_{\ell}\}$ is a basis of $\fql$ and $\{\bnu_1, \dots, \bnu_{\ell}\}$ is its trace dual basis.
    According to Definition \ref{def:linear_repair_scheme}, the linear repair algorithm computes coefficients $\eta_i \in \fq$ so that $f(\bals) = \sum_{i \in [\ell]}\eta_i\bnu_i$, where $\eta_i = \tr\big(\bbe_i f(\bals)\big)$ for $i \in [s]$, and  $\eta_i$'s, $i\in [s+1, \ell]$, are $\fq$-linear functions of the queries in 
    $\big\{\tr\big(\bga f(\bal)\big): \bga \in \cup_{\bal \in A\setminus\{\bals\}}Q_{\bal}\big\}$, i.e., for $i\in [s+1, \ell]$,
    \[
    \begin{split}
    \eta_i
    &= \sum_{\bal \in A\setminus\{\bals\}}\sum_{\bga \in Q_{\bal}}\om_{\bal, \bga}\tr\big(\bga f(\bal)\big)\\
    &= \sum_{\bal \in A\setminus\{\bals\}}\tr\big(\big(\sum_{\bga \in Q_{\bal}}\om_{\bal, \bga}\bga\big) f(\bal)\big)\\
    &= \sum_{\bal \in A\setminus\{\bals\}}\tr\big(\bmu_{\bbe_i,\bal} f(\bal)\big)\\
    &= \tr\big(\sum_{\bal \in A\setminus\{\bals\}}\bmu_{\bbe_i,\bal} f(\bal)\big),
    \end{split}
    \]
     for $\om_{\bal, \bga} \in \fq$,
     and $\bmu_{\bbe_i,\bal} \define \sum_{\bga \in Q_{\bal}}\om_{\bal, \bga}\bga \in \spn(Q_{\bal})$.
    Furthermore,
    $\tr\big(\beta_i f(\bals)\big) = \tr\big(\bbe_i\sum_{j \in [\ell]}\eta_j\bnu_j\big) = \sum_{j\in [\ell]}\eta_j\tr(\bbe_i\bga_j) = \eta_i$, $i \in [s+1, \ell]$.
    Then, for $i \in [s+1, \ell]$,
    \begin{equation}
        \label{eq:queries}
        \tr\big(\beta_i f(\bals)\big) = \tr\big(\sum_{\bal \in A\setminus\{\bals\}}\bmu_{\bbe_i,\bal} f(\bal)\big).
    \end{equation}
 The Equation \ref{eq:queries} holds for all polynomials $f \in \fql[x]$, $\deg(f) < k$, then for all $\bga \in \fqls$ and for all $i \in [s+1, \ell]$ it still holds, i.e., 
 $\tr\big(\bga\beta_j f(\bals)\big) = \tr\big(\bga\sum_{\bal \in A\setminus\{\bals\}}\bmu_{\bbe_i,\bal} f(\bals)\big)$, 
 which derives to
 $\beta_i f(\bals) = \sum_{\bal \in A\setminus\{\bals\}}\bmu_{\bbe_i,\bal} f(\bal)$. This completes the proof.
\end{proof}

 Lemma \ref{lem:scheme_queries_coefficients} ensures for the existence of a linear algorithm to repair a failed node $f(\bals)$ once there exists a linear repair scheme by ensuring the existence of the set $\{\bbe_i\}_{i \in [s+1, \ell]}$ and the elements $\bmu_{\bbe_i,\bal}$, for $i \in [s+1,\ell]$ and $\bal \in A\setminus\{\bals\}$. We now propose a linear algorithm to repair $f(\bals)$ of RS$(A,k)$ with side information $S = \{\bbe_i\}_{i\in [s]}$, which is modeled after the linear algorithm to repair RS$(A,k)$ in \cite{GuruswamiWootters2016}. 
\vspace{6pt}
\hrule
\begin{algorithm}
 \label{alg:linear_repair}
 Linear repair with side information $S = \{\bbe_i\}_{i \in [s]}$.
\hrule
\vspace{6pt}
\noindent
{\bf Input:} A set $A$ of evaluation points, a point $\bals \in A$ of the failed node $f(\bals)$, the $s$ traces $\tr\big(\bbe_i f(\bals)\big)$ corresponding to the side information $S = \{\bbe_i\}_{i \in [s]}$, the access to linear queries of the form $\tr(\bga f(\bal))$, for all $\bal \in A\setminus \{\bals\}$.\\
 {\bf Output:} the value $f(\bals)$.\\
 {\bf Steps:}
 \begin{enumerate}
     \item Choose a linearly independent set $B = \{\beta_i\}_{i\in [s+1, \ell]}$.
     \item \label{step2} Choose elements $\bmu_{\bbe_i,\bal} \in \fql$ for each pair of $\bbe_i \in B$ and $\bal \in A\setminus \{\bals\}$ so that 
     $\beta_i f(\bals) = \sum_{\bal \in A\setminus\{\bals\}}\bmu_{\bbe_i,\bal} f(\bal)$.
     \item for $\bbe_i \in B$ do
     \item Choose an arbitrary spanning set $Q_{\bal}$ for the set $\{\bmu_{\bbe_i, \bal}\}_{i \in [s+1,\ell]}$ and get the queries $\tr\big(\bga f(\bal)\big)$, $\bga \in Q_{\bal}$.
     \item Compute $\tr\big(\bmu_{\bbe_i,\bal}f(\bal)\big)$ for each $\bal \in A\setminus \{\bals\}$ through the traces $\tr(\bga f(\bal))$'s.
     \item Compute $\tr\big(\bbe_i f(\bals)\big) = \tr\big(\sum_{\bal \in A\setminus\{\bals\}}\bmu_{\bbe_i,\bal} f(\bal)\big),$ $i \in [s+1, \ell]$, by taking the trace of both sides of the equation in Step \ref{step2}.
     \item end
     \item Compute $f(\bals)$ from $\{\tr\big(\bbe_i f(\bals)\big)\}_{i\in [\ell]}$: 
     \[
     f(\bals) = \sum_{i \in [\ell]} \tr\big(\bbe_if(\bals)\big)\bnu_i,
     \]where $\{\bnu_1, \dots, \bnu_{\ell}\}$ is the dual basis of $\{\bbe_1, \dots, \bbe_{\ell}\}$.
 \end{enumerate} 
 \end{algorithm}
\hrule
\vspace{12pt}
    
The following proof of Proposition~\ref{pro:scheme_corresponds_polys} indicates that a linear repair scheme for a node with side information size $s$ of a code RS$(n,k)$ is equivalent to a set of $\ell-s$ polynomials of degree less than $n-k$.
\begin{proof}[Proof of Proposition~\ref{pro:scheme_corresponds_polys}]
Supposing that $\{\bnu_i\}_{i=1}^\ell$ is the trace-dual basis of $\{\bbe_i\}_{i=1}^\ell$, where $\{b_i\}_{i\in [s]} = S$ and $\{b_i\}_{i=s+1}^\ell = T$. Supposing that 
$f(\bals) = \sum_{i=1}^\ell \eta_i\bnu_i.$
According Lemma \ref{lem:scheme_queries_coefficients}, the work of defining $f(\bals)$ with side information $S$ is now the work of defining $\ell - s$ coefficients $\eta_i$, $i \in [s+1,\ell]$, and
$\tr\big(\bbe_if(\bals)\big) = \eta_i$.
This means that to define $\eta_i$, $i \in [s+1, \ell]$, it is enough to find $\tr\big(\bbe_if(\bals)\big)$, or $\tr\big(g_i(\bals)f(\bals)\big)$, $i \in [s+1, \ell]$.
Each polynomial $g_i(x)$, $i \in [s+1, \ell]$, of degree less than $n-k$ corresponds to a dual codeword of the Reed-Solomon codes RS$(A,k)$, which returns $\sum_{j=1}^ng_i(\bal_j)\lambda_jf(\bal_j) = 0$, where $\lambda = (\lambda_1,\dots, \lambda_n) \in \fql^n$.
Then, 
$g_i(\bals)\lambda^*f(\bals) = - \sum_{\bal_j \neq \bals}g_i(\bal_j)\lambda_jf(\bal_j)$, which is equivalent to
$g_i(\bals)f(\bals) = - \sum_{\bal_j \neq \bals}g_i(\bal_j)\frac{\lambda_j}{\lambda^*}f(\bal_j)$.
Applying the trace function on two sides of this equation we get 
$\tr\big(g_i(\bals)f(\bals)\big) = - \sum_{\bal_j \neq \bals}\tr\big(g_i(\bal_j)\frac{\lambda_j}{\lambda^*}f(\bal_j)\big)$.
In conclusion, each $\eta_i$, $i \in [s+1, \ell]$, can be computed through the traces $\tr\big(g_i(\bal_j)\frac{\lambda_j}{\lambda^*}f(\bal_j)\big), \bal_j \neq \bals$, which can totally be defined through the polynomials $g_i(x)$, $i \in [s+1, \ell]$.
\end{proof}

\subsection{Proof of Proposition~\ref{prop:lower_bound}}
\label{app:lower_bound}



\begin{proof}[Proof of Proposition~\ref{prop:lower_bound}]
According to Proposition~\ref{pro:scheme_corresponds_polys}, a linear repair scheme with side information size $s$ for a failed node $f(\bals)$ corresponds to a set of $\ell - s$ polynomials.
Supposing that the repair scheme  is the polynomial set $\{g_i(x)\}_{i\in [s+1,\ell]}$, where $\rank\big(\{g_{s+1}(\bals),\dots,g_\ell(\bals)\}\big) = \ell -s$ and 
$\rank\big(\{g_{s+1}(\bal),\dots,g_\ell(\bal)\}\big) = b_{\bal}$, for all $\bal\in A\setminus \{\bals\}$.
Therefore, the repair bandwidth of the repair scheme is $b = \sum_{\bal \in A\setminus\{\bals\}}b_{\bal}$. 
For each $\bal \in A$, let 
$S_{\bal} \define \{\vec{e} = (e_{s+1}, \dots, e_{\ell}) \in \fq^{\ell-s}: \sum_{i\in[s+1, \ell]}e_ig_i(\bal) = 0\}$.
Since $\rank\big(\{g_i(\bal)\}_{i \in [s+1, \ell]}\big) = b_{\bal}$, $\dim_{\fq}(S_{\bal}) = \ell - s - b_{\bal}$.
Averaging the size $|\{\bal \in A\setminus \{\bals\}: \vec{e} \in S_{\bal}\}|$ over all nonzero vectors $\vec{e} \in \fq^{\ell -s}$, we have
\[
    \begin{split}
    &\mu 
    \define \frac{1}{q^{\ell -s} -1}\sum_{\vec{e} \in \fq^{\ell - s} \setminus \{\vec{0}\}}\mid \{\bal \in A\setminus \{\bals\}: \vec{e} \in S_{\bal}\}\mid\\
    &=
     \frac{1}{q^{\ell -s} -1}\sum_{\bal \in A \setminus \{\bals\}}\mid \{\vec{e} \in \fq^{\ell - s} \setminus \{\vec{0}\}: \vec{e} \in S_{\bal}\}\mid\\
    &=
     \frac{1}{q^{\ell -s} -1}\sum_{\bal \in A \setminus \{\bals\}} \big(q^{\ell - s - b_{\bal}} -1\big).
    \end{split}
\]
Then, there exists some $\vec{e}^* = (e_{s+1}^*, \dots, e_{\ell}^*) \in \fq^{\ell - s}\setminus \{0\}$ such that 
$|\{\bal \in A\setminus \{\bals\}: \vec{e}^* \in S_{\bal}\}| \geq \mu$.
Let  $g(x) \define \sum_{i \in [s+1, \ell]}e_i^*g_i(x)$, $g(x)$ vanishes on at least $\mu$ elements of $A\setminus \{\bals\}$.
Furthermore, it follows from $\{g_i(\bals)\}_{i \in [s+1, \ell]}$ is linearly independent and $\vec{e}^* \neq 0$ that $g(\bals) \neq 0$. Hence, $g(x)$ corresponds to a nonzero dual codeword of RS$(A, k)$ and has at most $r-1$ roots, where $r = n-k$.
Then, $\mu \leq r - 1$,
which allows that

\begin{equation}
\label{eq:T}
\sum_{\bal \in A \setminus \{\bals\}}q^{-b_{\bal}} \leq \frac{(r - 1)(q^{\ell - s} - 1) + n - 1}{q^{\ell - s}}.
\end{equation}

Put 
\[
T \define \frac{(r - 1)(q^{\ell - s} - 1) + n - 1}{q^{\ell - s}},  \bave \define \log\frac{n-1}{T}.
\]

Let
\begin{equation}
\label{eq:b_min}
b_{\min} \define \min_{b_{\bal} \in \{0,\dots,\ell-s\}}\sum_{\bal \in A\setminus \{\bals\}}b_{\bal}
\end{equation}
subject to (\ref{eq:T}).

The minimum occurs when $b_{\bal}$'s are balanced and equal to $\bave$.
Supposing that $t$ is the biggest integer satisfying 
$$b_1^* = \dots = b_t^* = \lfloor \bave\rfloor, b_{t+1}^* = \dots = b_{n-1}^* = \lceil \bave\rceil,$$
where $\sum_{i \in [n-1]}q^{-b_i^*} \leq T$, and $(b_1^*, \dots, b_{n-1}^*)$ is an optimal solution for (\ref{eq:b_min}).
To obtain this solution, the $``$balancing$"$ procedure as in \cite{DauDinhKiahTranMilenkovic2021} is applied. The computation for $t$ is easily obtained. Then, we get the lower bound as desired.
\end{proof}

\subsection{A discussion on the subspace intersections with the lowest repair bandwidth}
\label{app:subspace_intersection_lowest_BW}
A condition for an $m$-dimensional subspace $\W$ so that the repair scheme constructed from this subspace by Lemma \ref{lem:bandwidth_intersection} and Corollary \ref{cr:full_length} obtains the minimal repair bandwidth among all $m$-dimensional $\fq$-subspaces is that the sum $\sum_{\bga \in \fqls}\dim(\bga\fqlms \cap \W)$ achieves the maximal value among all $m$-dimensional $\fq$-subspaces.
One concrete consideration for obtaining the maximal sum is the case when the intersection subspaces have dimension $0$ or $1$. More particularly, 
for a parameter $m$, if there exists an $m$-dimensional subspaces $\W_0$ with  $\dim(\bga\fqlms \cap \W_0) \in \{0,1\}$, for all $\bga \in \fqls$, then a sufficient condition for an arbitrary $m$-dimensional subspace $\W$ used to construct subspace polynomial repair scheme that obtains the lowest repair bandwidth, i.e., the sum $\sum_{\bal \in \fql\setminus\{\bals\}}\dim\big((\bal-\bals)\bde \fqlms \cap\W\big)$ achieves maximal value among all subspaces dimension $m$, is also the condition that $\dim(\bga\fqlms \cap \W) \in \{0,1\}$, for all $\bga \in \fqls$.
Moreover, for the codes RS$(n,k)$, where $n-k = q^m$, this condition is the necessary and sufficient condition for the repair scheme constructed by $\W$ obtaining the optimal repair bandwidth.
The repair schemes constructed in Theorems \ref{thm:bandwidth_coprime} and Theorem \ref{thm:greedy_scheme} are of this consideration. 
We will make the above discussion clearer in Corollary \ref{col:max_dim_with_01}.
Since our proof for the conclusion of subspace intersection of dimensions $0$ or $1$ is based on the \textit{majorization} of two real number sequences, we first recall some basic results on this problem. 
For two sequences of real numbers $x = (x_1,\dots,x_p)$ and $x' = (x'_1, \dots, x'_p)$,
supposing that $x_1\geq \dots \geq x_p$  and $x'_1\geq \dots \geq x'_p$, we say that $x$ is majorized by $x'$ or $x'$ majorizes $x$ if $\sum_{i=1}^px_i = \sum_{i=1}^px'_i$ and $\sum_{i=1}^j x_i\leq \sum_{i=1}^jx'_i$, for all $j\in [p-1]$ \cite[A.1, p.8]{MarshallOlkinArnold_2011}.

\begin{lemma}\cite[B.1, p.156]{MarshallOlkinArnold_2011}
    \label{lem:majorization_convex}
    The inequality $\sum_{i=1}^p \phi(x_i) \leq \sum_{i=1}^p\phi(x'_i)$ is satisfied for all continuous convex function $\phi: \mathbb{R} \rightarrow \mathbb{R}$ if and only if $x$ is majorized by $x'$.
    
\end{lemma}

Now we have the condition to get maximal value for the sum of intersection dimensions.
Supposing that $p \define \frac{q^\ell -1}{q^a -1}$, which is the number of disjoint cosets of $\fqas$ in $\fqls$. Since each pair of cosets are completely coincided or disjoint, and for all $\bga' \in \bga\fqas, \bga'\fqas = \bga\fqas$, we only need to consider the sums over $p$ disjoint cosets with the value of each dimension in each sum repeated $|\fqas| = q^a -1$ times.
We have the following proposition.
\begin{proposition}
   \label{pro:majorization_compare_sum}
 Let $\bga_1\fqas, \dots, \bga_p\fqas$ are $p$ disjoint cosets and the two sequences $d \define (d_1, \dots, d_p)$, $d' \define (d'_1, \dots, d'_p)$ are the dimensions of the intersection of subspaces $\bga_i\fqa$ of these cosets with $\W$ and $\W'$, respectively. Without lost of generality, we can suppose that 
$d_1\geq \dots \geq d_p$  and $d'_1\geq \dots \geq d'_p$. 
Let $x \define (x_1, \dots, x_p)$ and $x' \define (x'_1, \dots, x'_p)$, where $x_i \define q^{d_i}-1$ if $d_i > 0$ and $x_i \define 0$ if $d_i = 0$, and  $x'_i \define q^{d'_i}-1$ if $d'_i > 0$ and $x'_i \define 0$ if $d'_i = 0$. Then, if $x$ is majorized by $x'$ then $\sum_{i=1}^{p}d_i \geq \sum_{i=1}^{p}d'_i$.
\end{proposition}

\begin{proof}
Since $x$ is majorized by $x'$ and
$d_i = -\phi(x_i), d'_i = -\phi(x'_i)$, where $\phi(t) \define -\log_q(t+1)$ is a continuous convex function over $[0, +\infty)$.
The proof is completed by applying Lemma \ref{lem:majorization_convex} for $x$ and $x'$ and the computation of $d_i$ and $d'_i$ through function $\phi(t)$.
\end{proof}

\begin{corollary}
     \label{col:max_dim_with_01}
    Let $a|\ell$, $\W$ and $\W'$ are two $m$-dimensional subspaces of $\fql$ where $\dim(\bga\fqa\cap\W) \in \{0, 1\}$, for all $\bga \in \fqls$. Then, $\sum_{\bga\in \fqls}\dim(\bga\fqa \cap \W) \geq \sum_{\bga\in \fqls}\dim(\bga\fqa \cap \W')$.  Moreover, $\sum_{\bga\in \fqls}\dim(\bga\fqa \cap \W) = \frac{(q^a -1)(q^m -1)}{q-1}$, and if there exists $\bga' \in \fqls$ so that $\dim(\bga'\fqa \cap \W') > 1$ then $\sum_{\bga\in \fqls}\dim(\bga\fqa \cap \W) > \sum_{\bga\in \fqls}\dim(\bga\fqa \cap \W')$.
\end{corollary}

\begin{proof}
The proof is completed by applying Proposition \ref{pro:majorization_compare_sum} and Lemma \ref{lem:majorization_convex} for two sequences $x$, $x'$ in the special case where $d_i \in \{0,1\}$, for all $i \in [p]$.
When all the intersections $\bga\fqas \cap \W$ is of dimension $0$ or $1$, the set of $1$-dimensional intersections is a partition of $\W$, which allows that the number of these subspaces is $\frac{q^m - 1}{q-1}$. Since each of the intersection is repeated $q^a-1$ times, the total of dimensions is $\frac{(q^a -1)(q^m -1)}{q-1}$.
If there exists $d'_j >1$, for some $j \in [p]$, then the strict inequality is achieved.
\end{proof}

\end{document}